\documentclass[12pt,a4paper]{article}
\newcommand{\templatetitle}{Preference for Verifiability} 
\date{February 2026}
\title{\templatetitle}

\usepackage{authblk, varwidth, etoolbox}\NewCommandCopy{\oldauthor}{\author}\NewCommandCopy{\oldaffil}{\affil}\newcommand{\templateauthors}{}
\renewcommand{\author}[2][2]{%
  \ifdefempty{\templateauthors}
    {\appto\templateauthors{#2}}
    {\appto\templateauthors{, #2}}
    \oldauthor[#1]{#2}%
  }
  \renewcommand{\affil}[2][2]{\oldaffil[#1]{%
      \begin{minipage}[t]{12.5cm}\protect\footnotesize
        #2
      \end{minipage}%
    }}

\author[1,2]{Rommeswinkel, Hendrik}

\affil[1]{Graduate School of Economics, Hitotsubashi University \par 2-ch\={o}me-1 Naka, Kunitachi City, Tokyo 186-0004, Japan}
\affil[2]{Waseda Institute for Advanced Study, Waseda University \par 1-ch\={o}me-104 Totsukamachi, Shinjuku City, Tokyo 169-8050, Japan}

\newcommand{\templatesubject}{Decision Theory} 
\newcommand{\templatekeywords}{Verifiability, uncertainty, Choquet expected
  utility, greenwashing, principal-agent} 
\newcommand{\templatejelcodes}{D81, Q54, M14} 
\newcommand{\templatethanks}{\thanks{I thank presentation participants
    at SDM 2024, at the Hitotsubashi Summer Institute, at the Normative
    Economics and Public Policy Seminar, at the Karlsruhe Institute of
    Technology, at the Zukunftskolleg, University of Konstanz, at
    National Taiwan University, and at Kobe University for helpful
    comments.}}
    \title{\templatetitle\templatethanks}
\usepackage{titling}

\usepackage{sectsty}
\makeatletter\def\@seccntformat#1{\protect\makebox[0pt][r]{\csname the#1\endcsname\hspace{11pt}}}\makeatother
\allsectionsfont{\sc}
\newcommand{\printkeywords}{{\sc Keywords:}
  \templatekeywords}
\newcommand{\printjelcodes}{{\sc JEL Classification:}
  \templatejelcodes}

\makeatletter
\def\@maketitle{%
  \newpage
  \null
  \vskip 2em%
  \begin{center}%
    \let \footnote \thanks
    {\LARGE \@title \par}%
    \vskip 1.5em%
    {\large
      \lineskip .5em%
      \begin{tabular}[t]{c}%
        \baselineskip=12pt
        \@author
      \end{tabular}\par}%
    \vskip 1em%
    {\large \@date}%
  \end{center}%
  \par
  \vskip 1.5em}
\makeatother

\usepackage{setspace}
\usepackage[utf8]{inputenc}
\usepackage{csquotes}
\usepackage[english]{babel}

\usepackage{contour}
\usepackage[normalem]{ulem}

\contourlength{0.8pt}

\usepackage{scalerel}

\def\ImportFromMnSymbol#1{%
  \DeclareFontFamily{U} {MnSymbol#1}{}
  \DeclareFontShape{U}{MnSymbol#1}{m}{n}{
    <-6> MnSymbol#15
    <6-7> MnSymbol#16
    <7-8> MnSymbol#17
    <8-9> MnSymbol#18
    <9-10> MnSymbol#19
    <10-12> MnSymbol#110
    <12-> MnSymbol#112}{}
  \DeclareFontShape{U}{MnSymbol#1}{b}{n}{
    <-6> MnSymbol#1-Bold5
    <6-7> MnSymbol#1-Bold6
    <7-8> MnSymbol#1-Bold7
    <8-9> MnSymbol#1-Bold8
    <9-10> MnSymbol#1-Bold9
    <10-12> MnSymbol#1-Bold10
    7   <12-> MnSymbol#1-Bold12}{}
  \DeclareSymbolFont{MnSy#1} {U} {MnSymbol#1}{m}{n}
}
\newcommand\DeclareMnSymbol[4]{\DeclareMathSymbol{#1}{#2}{MnSy#3}{#4}}
\ImportFromMnSymbol{A}
\DeclareMnSymbol{\ConIndepNat}{\mathrel}{A}{225}

\usepackage[sc,osf]{mathpazo}
\usepackage{amsthm, amsmath, amssymb}
\usepackage{latexsym} 
\usepackage[mathscr]{euscript} 
\usepackage{microtype} 
\usepackage{mathtools} 
\usepackage{subdepth} 

\usepackage[backend=biber, style=apa, autocite=plain]{biblatex}
\DeclareLanguageMapping{english}{english-apa}
\AtEveryBibitem{%
  \ifentrytype{article}{
      \clearfield{url}%
      \clearfield{urldate}%
      \clearfield{isbn}%
      \clearfield{issn}%
    }{}
    \ifentrytype{book}{
        \clearfield{url}%
        \clearfield{urldate}%
        \clearfield{isbn}%
      }{}
      \ifentrytype{collection}{
          \clearfield{url}%
          \clearfield{urldate}%
          \clearfield{isbn}%
        }{}
        \ifentrytype{incollection}{
            \clearfield{url}%
            \clearfield{urldate}%
            \clearfield{isbn}%
          }{}
        }

\usepackage{graphicx}
\graphicspath{{../Graphs/}}
\usepackage{tikz}
\usepackage{comment}
\usetikzlibrary{graphs}
\usetikzlibrary{decorations.pathreplacing}

\usepackage{booktabs}
\usepackage{array}
\usepackage{float}
\usepackage{picinpar} 
\usepackage{multirow}

\usepackage{enumitem}

\usepackage[allcolors=black]{hyperref}
\hypersetup{
    pdftitle={\templatetitle},
    pdfauthor={\templateauthors},
    pdfsubject={\templatesubject},
    pdfkeywords={\templatekeywords},
    bookmarksnumbered=true,
    bookmarksopen=true,
    bookmarksopenlevel=1,
    colorlinks=true,
    pdfstartview=Fit,
    pdfpagemode=UseOutlines,    
    pdfpagelayout=TwoPageRight
}

\usepackage[margin=1.25in]{geometry}
\usepackage{tabularx}
\usepackage{ellipsis} 

{
{\theoremstyle{definition}\newtheorem{definition}{Definition}}
{\theoremstyle{remark}}
{\theoremstyle{definition}}
{\theoremstyle{definition}\newtheorem{axiom}{Axiom}}
\newtheorem{theorem}{Theorem}
\newtheorem{corollary}{Corollary}
\newtheorem{proposition}{Proposition}
\newtheorem{lemma}{Lemma}
}

\newcommand*{\myexamplename}{Example}
\newenvironment{example}[1][\myexamplename]{\begin{proof}[#1]}{\end{proof}}

\usepackage[title]{appendix}
\begin{document}

%
%

{\let\clearpage\relax%
\maketitle }
\thispagestyle{empty}
\clearpage
\begin{abstract}\noindent
  Decision makers sometimes cannot observe the consequences of their
  actions ex-post. This paper axiomatically characterizes a decision
  model in which the decision maker cares about verifying that a good
  consequence has been achieved. Preferences over acts identify a set
  of events the decision maker expects to verify. Decision makers
  choose acts maximizing, in expectation over verifiable events, the
  worst-case utility consistent with each event. A dual model captures
  decision makers who instead seek to obscure poor outcomes from
  verification. As an application, firms choosing carbon-reduction
  technologies may prefer less efficient but more verifiable
  technologies to prove emission reductions to stakeholders.

	\noindent\printkeywords

	\noindent\printjelcodes
\end{abstract} \bigskip

\thispagestyle{empty} \clearpage \setcounter{page}{1}
\newpage

\section{Introduction}\label{sec:introduction}
In many real-world decisions, the final consequences are only
partially observable. For instance, when a firm chooses a method for
CO2 emission reduction, the actual amount of CO2 removed from the
atmosphere is not directly observed but must be inferred from indirect
evidence, such as carbon compensation certificates. This lingering
ex-post uncertainty about consequences can motivate choices that
deviate from standard expected utility theory. A firm might not simply
aim to maximize the expected reduction in emissions; its incentives
are often more complex due to the presence of stakeholders. The firm
may want to {\em prove} that a significant emission reduction has been
achieved, or conversely, it may want to avert blame if stakeholders
can prove a desirable target was missed.

This paper's central theoretical contribution is to model these
motivations by providing an axiomatic foundation for preferences that
depend on the ex-post verifiability of outcomes. We ask: how does the
anticipation of what can be proven after the fact shape choice under
uncertainty? We show that behavior driven by a preference for or
against ex-post verifiability can be naturally understood through the
lens of ambiguity aversion models fulfilling comonotonic independence
\parencite{schmeidler_subjective_1989}, which have traditionally been
motivated by ex-ante uncertainty. We rationalize ambiguity
attitudes as a coherent response to unobservable consequences.

The applications we consider, such as corporate environmental claims,
are characterized by a lack of objective probabilities. This
subjective environment often gives rise to phenomena like
``greenwashing,'' where claims about expected emission reductions
greatly differ between firms and stakeholders. Consequently, we adopt
a fully subjective Savage framework with subjective mixtures
\parencite{ghirardato_subjective_2003} to identify risk preferences.

We propose two models that capture the principal-agent relationships
that can arise from imperfect ex-post verifiability using a dual-self
representation similar to \textcite{chandrasekher_dual_2022}. The
first model, {\em certification utility}, describes a decision-maker
who wants to prove a good outcome was achieved. They evaluate acts
based on the expected utility of the outcome they can prove, i.e., the
worst outcome consistent with their best proof. This captures the
behavior of a firm genuinely seeking to demonstrate its positive
environmental impact and corresponds to extreme ambiguity aversion.
The second model, {\em obfuscation utility}, captures the opposite
motive. Here, a stakeholder presents evidence, and the firm tries to
point to the most favorable interpretation consistent with that
evidence. This model provides a behavioral foundation for
greenwashing, where firms exploit ambiguity to avoid blame for poor
outcomes, and corresponds to extreme ambiguity seeking.

Our main representation theorems show that from the decision maker's
choices alone, an outside observer can identify the key elements of
the model: the presence of a hidden stakeholder, the nature of their
ex-post interaction (verification or obfuscation), and the set of
events the decision maker anticipates being verifiable along with the
probabilities of these events. This means we can distinguish between a
firm that is an expected utility maximizer, one that is genuinely
seeking verifiability, and one that is engaging in obfuscation, simply
by observing their preferences. The endogenously derived set of
verifiable events aligns with standard assumptions in contract theory,
being closed under intersections but not necessarily under
complements.

Our decision-theoretic results have policy implications for
corporate carbon responsibility (CCR). As policymakers seek to
regulate unverified environmental claims---for instance, by planning
to ``proscribe [...] generic environmental claims [...] without
proof'' and ``claims based on emissions offsetting schemes''
\parencite{council_ban_2023}---our framework provides a tool to
analyze the behavioral consequences. A welfare analysis reveals that
both certification-seeking and obfuscating behaviors can lead to
welfare losses. These losses are non-monotonic in the degree of
verifiability; that is, increasing transparency requirements may lower
welfare unless full ex-post verifiability is achieved. This suggests
that policies promoting CCR and transparency must be designed with
care, as their effectiveness depends critically on the set of actions
available to firms.

The paper proceeds as follows: Section \ref{sec:example} introduces
the example application of carbon offsetting programs in more detail.
Section \ref{sec:notation} presents the notation before Section
\ref{sec:models} discusses the two decision models. The axioms that
characterize these models are developed in Section \ref{sec:axioms}.
Comparative statics for the models can be found in Section
\ref{sec:comparative}. The paper also contributes a decision-theoretic
foundation to a model of greenwashing and on verifiability in contract
theory. Such relations to the literature are discussed in Section
\ref{sec:literature}.

\section{Introductory Example} \label{sec:example}

This section introduces an example of how the model presented in this
paper can capture certain stylized facts\footnote{In the U.N. Race to
  Zero Campaign \parencite{unfccc_race_2023}, an analysis of the
  largest firms of eight sectors \parencite{new_corporate_2023} found
  that most firms relied on renewable energy certificates to offset
  carbon emissions while only a small number of firms used
  certificates of nature-based carbon capture.
  \textcite{new_corporate_2023} concluded that ``Companies' climate
  pledges for 2030 fall well short of the required ambition and are
  inappropriately verified'' (p. 6), that ``the climate strategies of
  15 of the 24 companies [are] of low or very low integrity'' and that
  ``targets and potential offsetting plans remain ambiguous'' (p. 5).
} about how firms interact with their stakeholders in CCR issues where
beliefs are subjective and consequences are imperfectly observable,
such as carbon emission reductions. In such settings, subjective
beliefs (for example about the efficacy of the firm's actions) may
greatly differ between stakeholders and firms and there usually does
not exist an objective probability distribution to relate to. If firms
and stakeholders cannot agree ex-ante what the best course of action
is, this leaves only the ex-post outcomes for the stakeholder to
evaluate whether the firm has done well.

Suppose a firm decides between different carbon offsetting programs.
One option is to purchase certificates of nature-based carbon capture,
or more pointedly, to plant {\sc Trees}. Another option is to purchase
renewable energy certificates, {\sc RECs}, and the third option is to
improve the production processes to increase {\sc Efficiency}. For
simplicity, we assume that the total budget to be spent on the
offsetting programs is fixed.

Suppose the set of consequences $\mathscr{X} = \mathbb{R}_{+}$ is the amount of CO2
emissions avoided in megatons. Let the set of states $\mathscr{S} =
\{s,t,u\}$ consist of three mutually exclusive states. In state $u$,
there is a low supply of emission-reduction certificates (both RECs
and nature-based carbon removal). This makes purchasing such
certificates costly and the budget is only sufficient to offset a
small amount of CO2 emissions. In states $s$ and $t$, the supply of
such certificates is sufficient to offset a large amount. In state
$s$, purchasing RECs leads consumers in the energy market to
substitute from consuming a mix of sustainably and unsustainably
produced electricity to a higher share of unsustainably produced
energy. In state $t$, no such shift in behavior occurs. The options
are shown in Figure \ref{fig:carbon}.

\begin{figure}[ht]
	\centering
	\begin{tabularx}{.4\linewidth}[ht]{l|ccc}
		Act        & $s$ & $t$ & $u$ \\\hline
		Trees      & 70  & 70  & 10  \\
		RECs       & 60  & 100 & 10  \\
		Efficiency & 40  & 40  & 40
	\end{tabularx}
	\caption{Carbon Reduction Programs}
	\label{fig:carbon}
\end{figure}

Which of the three programs should a firm choose? According to
expected utility, this would depend on the risk preferences and the
subjective beliefs of the decision maker. A sufficiently high
subjective probability of $s$, $t$, or $u$ can induce any of the three
acts {\sc Trees}, {\sc RECs}, or {\sc Efficiency}, respectively, to be
optimal.

Suppose for a moment that there is no way of observing (or proving)
the exact amount of CO2 reduction achieved by the different acts
because only the total amount of emissions from many polluters can be
measured in the atmosphere. In the absence of objective probabilities
over states and when consequences are completely unverifiable, we can
think of two stereotypical behaviors of firms when interacting with
their stakeholders about carbon emissions:

A greenwashing firm may choose to purchase {\sc RECs} and may
exaggerate in their communication to stakeholders the likelihood that
state $t$ obtains. For example,
\textcite{new_corporate_2023} claim that for most firms joining the
U.N. race to zero campaign \parencite{unfccc_race_2023}, the CO2
emission reductions are not sufficiently verified. Firms may choose to
purchase {\sc RECs} even if they subjectively believe state $t$ to be
very unlikely. Such behavior can be detected by an analyst from
preferences under which uncertain acts are indifferent to their best
possible consequence.

A certification-seeking firm may choose {\sc Efficiency} to eliminate
any ex-post doubt of their stakeholders that they have achieved a
certain amount of CO2 emission reductions. For example, a firm may try
to prove to its stakeholders that its products are carbon neutral.
Such firms may choose {\sc Efficiency} even if they believe state $u$
to be very unlikely because ex post they cannot exclude this state. If
however the event $\{s,t\}$ would be ex-post verifiable, they may
choose {\sc Trees} instead. Such behavior can be detected from
preferences that treat uncertain acts indifferent to their worst
possible consequence.

In the following, we provide two models ---and axiomatic
characterizations thereof--- that correspond to these two types of
behavior. Compared to the extreme cases discussed above, we allow for
some events, for example $\{s,t\}$, to be ex-post verifiable and
subjectively determined. From an applied perspective, the main
contribution is to provide conditions on preferences over emission
reduction methods from which greenwashing behavior can be
distinguished from certification-seeking or expected utility behavior.
In addition, for any firm following either preference, we can identify
what information they expect to become available in the future.

\section{Notation}
\label{sec:notation}
Let $\mathscr{X}$ be a set of {\em consequences} and $\mathscr{S}$ a
finite set of {\em states of the world}. $\mathscr{E} \subseteq
2^{\mathscr{S}}$ denotes the sigma algebra of {\em events}. A {\em
  simple act} is a measurable function $a : \mathscr{S} \rightarrow \mathscr{X}$
with a finite image. If $E$ is an event and $a,b$ are acts, then
$a_Eb$ is the act that agrees with $a$ on all states $s \in E$ and that
agrees with $b$ on all states $s \in \overline{E} \equiv \mathscr{S} - E$. If
$x \in \mathscr{X}$, then $x$ also denotes the constant act that implements $x$ in
all states.

We analyze preferences over acts. A {\em preference relation} is a
binary relation $\succsim$ on $\mathscr{A}$. A function $U: \mathscr{A} \rightarrow \mathbb{R}$
{\em represents} $\succsim$ if $U(a) \geq U(b) \Leftrightarrow a \succsim b$. A representation $U$ is
{\em monotonic} if $U(a) \geq U(b)$ whenever for all $s \in \mathscr{S}$,
$a(s) \succsim b(s)$. The {\em certainty equivalent} of an act $a$, denoted
$[a] \in \mathscr{X}$, is a consequence such that $[a] \sim a$.

An event $E$ is said to be {\em null} if $\beta \sim \gamma E \beta$ for all $\gamma,\beta \in \mathscr{X}$
such that $\gamma \succ \beta$. An event $E$ is said to be {\em universal} if $\gamma \sim
\gamma E \beta$ for all $\gamma,\beta \in \mathscr{X}$ such that $\gamma \succ \beta$. An event $E$ is said to be
{\em essential} if $\gamma \succ \gamma E \beta \succ \beta$ for some $\gamma,\beta \in \mathscr{X}$.

Unlike in expected utility, in our axiomatization what happens on null
events may in principle still be preference relevant. We therefore
introduce the stronger notion of irrelevant events. An event $E$ is
said to be {\em irrelevant} if $\gamma_Ea \sim a$ for all $\gamma \in \mathscr{X}$
and $a \in \mathscr{A}$. A state of the world is said to be irrelevant
if it is an element of an irrelevant event. The set of events that are
not irrelevant is denoted $\mathscr{E}^*$ and the set of states that
are not irrelevant is denoted $\mathscr{S}^*$. An event that is null
need not be irrelevant. In our model, an event is only irrelevant for
the decision maker's preference if it is irrelevant in this stronger
sense we defined here.

If $\mathscr{F}$ is a set of events then $f: \mathscr{F}\rightarrow \mathbb{R}$ is called
a {\em set function}. A set function is {\em grounded} if $f(\emptyset) = 0$.
A set function is {\em normalized} if $f(\mathscr{S})=1$. A set
function that is both normalized and grounded is called a {\em
  capacity}. A set function is {\em additive} if for all $E,F \in
\mathscr{F}$ such that $E \cap F = \emptyset$, $f(E) + f(F) = f(E \cup F)$. A {\em
  probability measure} is a grounded, normalized, additive set
function. A set function is called {\em supermodular (submodular)} on
a set of events $\mathscr{G}$ if for all $E,F \in \mathscr{G}$, if $\eta(E)
+ \eta(F) \leq (\geq) \eta(E \cup F) + \eta(E \cap F)$. A set function is {\em modular} if
it is both supermodular and submodular.

\section{Decision Models} 
\label{sec:models}
We analyze two decision models. In the first decision model, the
decision maker wants to verify that a particular utility
level has been reached. In the context of our example, this may arise
because the decision maker wants to prove to a stakeholder that a
certain benefit of taking the action has been achieved.

\begin{definition}[Expected Certification Utility]
	A preference relation $\succsim$ on $\mathscr{A}$ is an expected certification utility
	if there exists a set of verifiable events $\mathscr{V} \subseteq \mathscr{E}$, closed under
	intersections and containing $\mathscr{S}^*$, a probability measure $\mu: \mathscr{E} \rightarrow
		[0,1]$, and a utility function $u: \mathscr{X} \rightarrow \mathbb{R}$ such that
	\begin{align}
		\label{eq:VerificationUtility}
		U(a) = & \int_{s \in \mathscr{S}^*} \max_{E \in \mathscr{V}: s \in E} \min_{t \in E} u(a(t)) d\mu
	\end{align}
	represents $\succsim$.
\end{definition}
$\mathscr{V}$ is called the set of {\em verifiable events}. If
$\mathscr{V} = \mathscr{E}$, then every state is verifiable and the
decision maker maximizes expected utility. The verifiable events do
not necessarily form a partition but are closed under intersections
and thus form a $\pi$-system that contains $\mathscr{S}^{*}$. The
interpretation of a verifiable event $E$ is that if the event obtains
(i.e., if the true state of the world is within $E$), then the
decision maker receives a proof that they can use to prove that this
event obtained. They then use this proof to convince the stakeholder
that at least the consequence $\min_{t \in E}u(a(t))$ has been
achieved.\footnote{The model assumes the verifiable events to be
  act-independent. The act-dependent case is not treated in the
  current paper for two reasons: first, if such verifiable events
  would be objectively given for every act, then this would make the
  analysis almost trivial. Second, if there are subjectively
  determined act-dependent verifiable events, then almost any behavior
  is possible and the verifiable events are unlikely to have
  interesting uniqueness properties.}

The decision model is a special case of an ambiguity aversion model
within the intersection of Choquet expected utility and maxmin
expected utility. However, the motivation for deviating from expected
utility is not ambiguity aversion but the lack of ex-post
verifiability of events: across verifiable events, the decision maker
follows expected utility but whenever consequences arise on events that
are not verifiable, the decision maker is extremely pessimistic.

\begin{example}
  Firm A would like to be able to prove to its stakeholders that a
  high level of reduction of CO2 emissions has been achieved. The firm
  believes that states $s$ and $u$ are unlikely, though not impossible.
  In case the event $\{s,t\}$ obtains, the firm believes it is able to
  prove this to its stakeholders. For example, after purchasing any
  amount of certificates of any kind, the firm can show these
  certificates to its stakeholders. However, the firm is unable to
  prove that $\{u\}$ obtains if certificate prices are not publicly
  observable -- it cannot use a low number of purchased certificates to
  prove that it is impossible to buy more certificates. The firm is
  also unable to prove that $\{s\}$ or that $\{t\}$ obtains since
  proving this would require observing the hypothetical behavior of
  how other market participants would behave given a purchase of {\sc
    RECs} or no purchase of {\sc RECs}. The firm can however verify
  that $\{s,t,u\}$ obtains since there is no other nonnull state. Firm
  A therefore considers two events when making a decision, $\{s,t\}$
  and $\{s,t,u\}$. With probability $\mu(\{u\})$ it can only prove that
  $\{s,t,u\}$ obtains and that the emission reduction is at least
  equal to the worst consequence of an act. With probability $\mu(\{s,t\})$
  the firm can verify that $\{s,t\}$ obtains and that the emission
  reduction is at least equal to the worst possible consequence on states
  $s$ and $t$. The firm therefore multiplies the probability of $u$
  with the utility of the worst possible consequence on $\{s,t,u\}$ and
  the probability of $\{s,t\}$ with the worst possible consequence on
  $\{s,t\}$. If the probability of state $u$ is sufficiently low, the
  firm will choose {\sc Trees}. If $u$ is sufficiently likely, the
  firm will choose {\sc Efficiency}.
\end{example}

In the second decision model, the decision maker has the desire to
obfuscate whether a bad consequence might have resulted. In the context of our
example, this may arise because a stakeholder may confront the firm
with evidence about the state of the world (e.g., that $\{s\}$ obtains
after having chosen {\sc RECs}).

\begin{definition}[Expected Obfuscation Utility]
	A preference relation $\succsim$ on $\mathscr{A}$ is an expected obfuscation utility
	if there exists a set of verifiable events $\mathscr{V}$, closed under
	intersections, a unique probability measure $\mu: \mathscr{E} \rightarrow [0,1]$, and a
	unique utility function $u: \mathscr{X} \rightarrow \mathbb{R}$ such that
	\begin{align}
		\label{eq:ValidationUtility}
		U(a) = \int_{s \in \mathscr{S}^{*}} \min_{E \in \mathscr{V}: s \in E} \max_{t \in E} u(a(t)) d\mu
	\end{align}
	represents $\succsim$.
\end{definition}
The interpretation of this model is that after a stakeholder confronts
the decision maker with the proof that an event $E$ obtains, the
decision maker points to the best possible consequence that might have
been achieved on $E$. We assume that the stakeholder has adversarial
preferences toward the decision maker, i.e., the stakeholder tries to
prove that at most a certain amount of utility has been obtained. The
stakeholder always uses all available verifiable information (chooses
the smallest event in $\mathscr{V}$ that contains the true state) and
the decision maker then always points to the best consequence on this
event. This model is dual to the certification utility in the sense
that for a certification utility {\em minimizing} decision maker there
exists an expected obfuscation utility with a reverse preference over
outcomes and vice versa.

\begin{example}
	Firm F wants to evade negative publicity about its CO2 emissions. It
  worries that after making their choice, some stakeholder confronts
  them with evidence which state of the world obtains and proving to
  them that a bad consequence has resulted from their chosen act. Firm F
  wants to make sure that given the evidence they might be confronted
  with, there is a state of the world consistent with this evidence in
  which their chosen action yields a good result. The firm expects
  that only state $\{s,t\}$ can be proven by the stakeholder (for
  example by low market prices for RECs.) However, the firm does not
  expect anyone to be able to obtain definitive proof that either
  $\{s\}$ or $\{t\}$ obtains. If $\{s,t\}$ is sufficiently likely, the
  firm chooses {\sc RECs} in order to be able to point to the good
  state $t$ in which 100 megatons of CO2 have been reduced. However,
  if state $\{u\}$ is sufficiently likely, the firm will choose to
  improve its {\sc Efficiency}.
\end{example}

Both of these decision models are extreme cases and unobservable
consequences may generate a rich variety of decision models that are less
extreme. The present decision models only represent a starting point
for the exploration of decision theories in which consequences cannot be
observed. The following axioms are therefore to be understood as a
starting point for a systematic analysis of preferences under
uncertainty with imperfect ex-post verifiability of consequences.

\section{Axioms}
\label{sec:axioms}
This section provides a set of axioms that characterize expected
verification and obfuscation utilities. The axioms allow us to
distinguish between verification and obfuscation seeking behavior and
should be normatively compelling if consequences are not observable.

We assume the existence of a biseparable utility representation.
Biseparable preferences allow us to identify the decision weights of
individuals on events while making minimal assumptions about behavior.
The axioms for biseparable utility are provided in
\textcite{ghirardato_risk_2001} and for convenience restated in
appendix \ref{app:biseparable}.

\setcounter{axiom}{-1}
\begin{axiom}[Biseparable Preference]\label{axiom:biseparable}
  $\succsim$ is a {\em biseparable preference} if there exists a monotonic
  representation $U: \mathscr{A} \rightarrow \mathbb{R}$, an essential event $E \in
  \mathscr{E}$, a set function $\mu: \mathscr{E} \rightarrow [0,1]$, such that for
  all $\gamma \succsim \beta$ and all events $F \in \mathscr{E}$:
  \begin{align}
    \label{eq:BiseparableUtility}
    U(\gamma F \beta) = \mu(F) U(\gamma) + (1-\mu(F)) U(\beta)
  \end{align}
  The image $U(\mathscr{X})$ is a convex set.
\end{axiom}

Thus, in the biseparable model, decision makers' preferences are only
meaningfully restricted on binary acts. Under the assumption that
certainty equivalents exist and preferences are biseparable, we can
define preference averages of consequences and acts.
These preference averages, first introduced in
\textcite{ghirardato_subjective_2003}, allow us to use
\textcite{anscombe_definition_1963} style axioms without resorting to
objective probabilities. While objective probability
lotteries can easily be implemented in a laboratory setting, for
applications in the context of greenwashing behavior, such lotteries
are often unavailable.

\begin{definition}[Preference Average]
	For all $x,y \in \mathscr{X}$ with $x \succsim y$, $z$ is a {\em preference
    average} of $x$ and $y$ if for some essential event $E$, $xEy \sim
  [{xEz}]E[{zEy}]$. $z$ is denoted by $1/2 x \oplus 1/2 y$. For two acts
  $a,b \in \mathscr{A}$, $c$ is a {\em pointwise preference average} of
  $a$ and $b$ if for all $s \in \mathscr{S}$, $c(s) = 1/2 a(s) \oplus 1/2
  b(s)$. $c$ is denoted by $1/2 a \oplus 1/2 b$.
\end{definition}

\begin{example}
  Suppose {\sc Trees} and {\sc RECs} refer to investing a fixed amount
  of money into the relevant carbon offset certificates. Suppose
  further that there are no returns to scale to either technology. If
  the decision maker is risk-neutral (in terms of the curvature of
  $u$) over the amount of CO2 emissions reduced, then $1/2$ {\sc
    Trees} $\oplus 1/2$ {\sc RECs} refers to the act in which the carbon
  emissions in every state of the world are the arithmetic average of
  {\sc Trees} and {\sc RECs}. This can be thought of as investing half
  the money into each technology. However, if the decision maker is
  not risk-neutral, this is not the case. Instead, in each state $1/2$
  {\sc Trees} $\oplus 1/2$ {\sc RECs} would offer the certainty equivalent
  of a (perfectly verifiable) equal probability lottery of the
  reduction achieved by {\sc Trees} and {\sc RECs} on that state.
\end{example}

Our first axiom, comonotonic independence, is a separability condition
from \textcite{schmeidler_subjective_1989} that is normatively
appealing in our context.
\begin{definition}[Comonotonic Acts]
	Acts $a,b \in \mathscr{A}$ are {\em comonotonic} if for all $s,s' \in \mathscr{S}$,
	\begin{itemize}
		\item $a(s) \succ a(s') \Rightarrow b(s) \succsim b(s')$, and
		\item $b(s) \succ b(s') \Rightarrow a(s) \succsim a(s')$.
	\end{itemize}
\end{definition}
Thus, two acts are comonotonic if they agree on the ranking of states
according to whether they achieve more desirable consequences. For
instance, two acts are not comonotonic if one yields a better
consequence in state $s$ than in $s'$, while the other yields a worse
consequence in $s$ than in $s'$. Constant acts are comonotonic with
all acts.

\begin{axiom}[Comonotonic Independence]
	$\succsim$ fulfills {\em comonotonic independence} if for all comonotonic $a,b,c$,
	$a \succsim b$ holds if and only if also
  \begin{align}
    \label{eq:ComonotonicIndependence}
    1/2 a \oplus 1/2 c \succsim 1/2 b \oplus 1/2 c.
  \end{align}
\end{axiom}
Comonotonic independence is normatively compelling in our setting
because deviations from expected utility should only arise from a lack
of verifiability. A preference reversal in
\eqref{eq:ComonotonicIndependence} should only occur if mixing with
act $c$ makes a good consequence verifiable for one act but not the
other. This is ruled out when all three acts are comonotonic. If
$a,b,c$ are comonotonic, then for any verifiable event, the state with
the worst outcome is the same for all three acts. Thus, mixing with
$c$ does not asymmetrically change what can be verified. Analogously,
also what can be obfuscated does not change: the best possible state
in each verifiable event will remain the same.

The following result, a corollary of M\"{o}bius inversion of the
capacity \parencite[e.g.,][]{chateauneuf_characterizations_1989},
makes precise why comonotonic independence is the right starting point
for modeling verifiability in decisions under uncertainty:

\begin{corollary}[Expected Verifiability
  Representation]\label{coro:ceu}
  The following statements are equivalent:
  \begin{enumerate}
  \item $\succsim$ is a biseparable preference fulfilling comonotonic
    independence.
  \item $\succsim$ can be represented by
	\begin{align}
		\label{eq:VerifiabilityUtility}
		U(a) = & p \int_{E \in \mathscr{E}^+} \min_{s \in E} u(a(s)) dp^+ \nonumber\\
           & + (1-p) \int_{E \in \mathscr{E}^-} \max_{s \in E} -u(a(s)) dp^-
	\end{align}
  where $\mathscr{E}^+$, $\mathscr{E}^-$ are disjoint subsets of
  $\mathscr{E}$, $p^+$ is a strictly positive probability measure over
  ex-post information sets $\mathscr{E}^+ $ and $p^-$ is a strictly
  positive probability measure over $\mathscr{E}^-$ and $u:
  \mathscr{X} \rightarrow \mathbb{R}$ a convex valued function.
  \end{enumerate}
\end{corollary}

This representation is effectively Choquet Expected Utility
\parencite[CEU;][]{schmeidler_subjective_1989}. That is, a
CEU-maximizing decision maker can be interpreted as maximizing an
expectation over possible ex-post information sets. On each
information set, the decision maker faces either a verification or an
obfuscating situation: if the event belongs to $\mathscr{E}^+$, the
worst consequence on that event is taken as given; if it belongs to
$\mathscr{E}^-$, the best consequence is taken as given. Thus, with
probability $p$, the decision maker expects a verification
situation---needing to prove that a good outcome occurred---with $p^+$
governing the distribution over what will be verifiable. With
probability $1-p$, they expect an obfuscating position, with $p^-$
governing the distribution over ex-post available verifiable events
and a reversed ranking of consequences. Because of this reversal of
preferences over consequences, we immediately see that comonotonic
independence together with a fixed ordering of outcomes requires the
decision maker to either be obfuscation-seeking or
certification-seeking, but not both.

While the representation nests our two main models as special cases,
it imposes no structure on which events belong to $\mathscr{E}^+$ and
which to $\mathscr{E}^-$. In particular, disjoint events
$E, F \in \mathscr{E}^+$ may have their union $E \cup F \in
\mathscr{E}^-$: the decision maker pessimistically evaluates $E$ and
$F$ individually yet optimistically evaluates any act whose outcome is
only known to lie in $E \cup F$. This is difficult to reconcile with
a coherent model of verifiability---what can be proven ex-post should
be determined by the state that occurred, not by an independent draw
from $p^+$ and $p^-$ that lives outside the state space. The axioms
we impose next restore this coherence.

We therefore now impose axioms on the structure of what is verifiable
ex-post. Our main identification tools will be {\em sensitive events}
and {\em reactive events}, that help us identify what the decision
maker expects to be ex-post verifiable: If a decision maker only cares
about the good consequences that they can verify ex-post, they will
often be unresponsive to improvements of acts on nonnull states if
these improvements are not verifiable. Conversely, if on an event the
decision maker dislikes a worsening of the consequence on any state,
then the decision maker must expect that if the event occurs, they
will be for sure in possession of {\em some} proof that the event (or
a subevent) obtained. We thus define:
\begin{definition}[Sensitive event]
	Let $\gamma \succ \beta$. An event $E$ is {\em sensitive } if $\gamma_E\beta \succ
  \gamma_{E-F}\beta$ for all nonnull events $F \subset E$.
\end{definition}
Thus, an event is sensitive if getting a worse outcome on any nonnull
subevent decreases preference. For an expected utility maximizer, all
nonnull events are sensitive. In certification utility, an event $E$
is sensitive if and only if the verifiable subsets of $E$ form a cover
of $E$---that is, every state in $E$ belongs to at least one verifiable
subevent of $E$.

In obfuscation utility, when a stakeholder confronts the decision
maker with evidence that event $E$ obtained, the decision maker
responds by pointing to the {\em best} possible consequence within
$E$. Hence, even a single good outcome anywhere in $E$ provides a
defense---a strictly bad consequence can only be proven if every state
in $E$ yields a bad outcome. As a result, the decision maker strictly
benefits from any improvement on any nonnull subevent of $E$. We thus
define:

\begin{definition}[Reactive event]
	Let $\gamma \succ \beta$. An event $E$ is {\em reactive } if $\beta_{E-F}\gamma \succ
  \beta_E \gamma $ for all nonnull events $F \subset E$.
\end{definition}
A reactive event therefore requires only an improvement on any nonnull
event and this increases the preference of the decision maker. Again, in
expected utility all events are reactive. In our obfuscation utility,
reactive events are the events $E$ for which the stakeholder has a
cover of verifiable events in $\mathscr{V}$ that are all subsets of
$E$.

For example, if $E - \{s\}$ is verifiable but $E$ itself is not, a
good outcome in state $s$ never helps the decision maker. When the
true state is $s$, any verifiable event the stakeholder can use must
extend beyond $E$, so the decision maker can already point to a good
outcome outside $E$. When the true state is in $E - \{s\}$, state $s$
does not appear in the evidence at all.

The decision maker can combine proofs of events. This generates a
combined proof that the intersection of the two events has occurred.
This is the standard assumption of verifiability in contract theory
\parencite{bull_evidence_2004}.
\begin{axiom}[Proof Conjunction]
  $\succsim$ fulfills {\em proof combinations with respect to sensitive
    (reactive) events} if for all sensitive (reactive) events $E$
  and $F$, the event $E \cap F$ is a sensitive (reactive) event.
\end{axiom}
To be precise, sensitive events are events for which covers of
potentially many proofs exist. However, if both $E$ and $F$ have such
a cover of proofs, then we can form a cover for $E \cap F$ by
intersecting every verifiable event in the cover of $E$ with every
verifiable event in the cover of $F$.

\begin{example}
  Suppose the firm can verify both $\{s,t\}$ (by presenting purchase
  certificates, showing certificate supply was sufficient) and
  $\{t,u\}$ (via an academic study showing that purchasing RECs in
  this market does not cause energy-market demand shifts). Proof
  Conjunction requires that $\{s,t\} \cap \{t,u\} = \{t\}$ is also a
  sensitive event. The firm can present both proofs simultaneously
  when state $t$ is the true state: the certificate proof of $\{s,t\}$
  rules out state $u$, while the study proof of $\{t,u\}$ rules out
  state $s$. Together, they certify that the true state is precisely
  $\{t\}$, where {\sc RECs} achieves 100 megatons compared to 70 for
  {\sc Trees}. If state $t$ is sufficiently likely, this additional
  proof therefore makes {\sc RECs} the preferred choice over {\sc
    Trees}.
\end{example}

The final key axiom makes precise that the decision maker is an
expected utility maximizer under the condition that every consequence
is verifiable. We impose this via ambiguity neutrality from the
literature on (ex-ante) ambiguity attitudes.
\begin{axiom}[Proof Sufficiency]
  $\succsim$ fulfills {\em proof sufficiency with respect to sensitive
    (reactive) events} if for all sensitive (reactive) events $E$ and
  $F$, all events $A \subseteq E \cup F$, and all consequences $\gamma \succsim \beta$, $\gamma' \succsim \beta'$
  if $f = \gamma_{A \cap E}\beta \sim \gamma'_{A \cap F} \beta'=g$ then $1/2 f \oplus 1/2 g \sim f$.
\end{axiom}

This axiom states that if an event is covered by two sensitive events
$E$ and $F$, then the bets on $A \cap E$, i.e., act $f$ and bets on $A \cap
F$, i.e., act $g$, are ambiguity neutral, i.e., there is no motive for
hedging; the mixed act $1/2 f \oplus 1/2 g$ is indifferent to both $f$ and
$g$.

\begin{example}
  Again suppose both $E = \{s,t\}$ and $F = \{t,u\}$ are sensitive.
  Consider first $A = \{s,t,u\}$, so $f = \gamma_{\{s,t\}}\beta$ and $g =
  \gamma'_{\{t,u\}}\beta'$. The good outcomes of both acts fall on verifiable
  events: the certificate proof establishes $\gamma$ for $f$ whenever $s$
  or $t$ obtains, and the study proof establishes $\gamma'$ for $g$
  whenever $t$ or $u$ obtains. The consequences of both acts are fully
  verifiable in every state and are therefore evaluated like
  expected-utility acts and mixing these acts should provide no benefit.

  Now consider $A = \{s,u\}$, so $f = \gamma_{\{s\}}\beta$ and $g =
  \gamma'_{\{u\}}\beta'$: the good outcomes now fall on non-verifiable states and
  indifference of the two acts requires $\beta \sim \beta'$. Under $f$, the firm
  achieves $\gamma$ in state $s$, but the certificate proof of $\{s,t\}$ only
  certifies the worst outcome on $E$, which is $\beta$ in state $t$. The
  good outcome is not verifiable. The same applies to $g$ in state
  $u$: the study proof of $\{t,u\}$ certifies at most $\beta'$, again
  because state $t$ yields $\beta'$ under $g$. State $t$ is a bottleneck
  in both proofs. Mixing the two acts does not resolve this: state $t$
  still yields a bad outcome $1/2\beta \oplus 1/2\beta'$.
\end{example}

\subsection{Characterization of Expected Certification Utility}
The above stated axioms contain necessary and sufficient conditions on
preferences for an expected certification utility. More precisely, we
obtain the following characterization result:

\begin{theorem}[Representation Theorem]\label{thm:verification}
  The following statements are equivalent:
  \begin{enumerate}
  \item $\succsim$ is a biseparable preference fulfilling comonotonic
    independence, proof conjunction with respect to sensitive events, and
    proof sufficiency with respect to sensitive events.
  \item $\succsim$ is an expected certification utility.
  \end{enumerate}
\end{theorem}

The representation has standard uniqueness properties with respect to
$U$ and $u$. The interesting uniqueness properties are with respect to
the set of verifiable events and subjective beliefs $\mu$. Can we
identify from behavior which events the decision maker expects to be
able to verify? It turns out that while the set $\mathscr{V}$ is not
unique, its smallest elements are unique and once we close
$\mathscr{V}$ under unions, the resulting set is unique and so are the
beliefs over elements of this set and complements thereof:

\begin{proposition}[Uniqueness]\label{prop:uniqueness}
  Suppose $\succsim^1$ and $\succsim^2$ are expected certification utilities.
  Then $a \succsim^1 b \Leftrightarrow a \succsim^2 b$ for all $a,b \in \mathscr{A}$ if and only if:
  \begin{itemize}
  \item $U^1 = \theta U^2 + \phi$,
  \item $u^1 = \theta u^2 + \phi$,
  \item $cl_{\cup}(\mathscr{V}^1) = cl_{\cup}(\mathscr{V}^2)$,
  \item $\forall E \in cl_{\cup}(\mathscr{V}^1): \mu^1(E) = \mu^2(E)$ and
    $\mu^1(\overline{E}) = \mu^2(\overline{E})$.
  \end{itemize}
  for some $\theta \in \mathbb{R}_{+}$ and some $\phi \in \mathbb{R}$
\end{proposition}
\begin{example}
  In our example, the verifiability of $\{s, t\}$ and $\{s, t, u\}$
  can be inferred from behavior of the decision maker and need not be
  exogenously given. Moreover, the subjective probabilities of $\{s,
  t\}$ and $\{u\}$ are uniquely determined from behavior.
\end{example}

Effectively, when preferences over acts are identical then if one
decision maker can prove that a certain event obtains, the other
decision maker is able to prove that this event obtains or even able
to prove that a strict subset of this event obtains. Thus, it could be
the case that decision maker 1 is in possession of proofs that $E, F,
E \cup F, E \cap F$ obtain (whenever they actually obtain) while decision
maker 2 is only able to prove that $E, F, E \cap F$ obtain. In terms of
behavior, we would not see a difference because the decision maker
will always want to provide the most precise proof possible. Thus, we
cannot identify each proof uniquely but only the collection of all
proofs by closure under unions, the sensitive events. Of course, this
corresponds to a unique minimal cover of the state space with
verifiable events.

\subsection{Characterization of Expected Obfuscation Utility}
We also obtain a dual theorem that characterizes obfuscation utility.
Since an expected obfuscation utility is equivalent to an expected
certification utility with a reverse order of preferences,
supermodularity is replaced by submodularity and proof sufficiency
holds with respect to reactive instead of
sensitive events.
\begin{theorem}[Representation Theorem]\label{thm:obfuscation}
  The following statements are equivalent:
  \begin{enumerate}
  \item $\succsim$ is a biseparable preference fulfilling comonotonic
    independence, proof conjunction with respect to reactive events, and
    proof sufficiency with respect to reactive events.
  \item $\succsim$ is an expected obfuscation utility.
  \end{enumerate}
\end{theorem}

Naturally, the expected obfuscation utility has the same uniqueness
properties as expected certification utility. Expected obfuscation
utility can be interpreted as the decision maker minimizing the
utility of a stakeholder with expected certification utility who has a
reverse order of preference of consequences. Similarly, expected
certification utility can be seen as minimizing the utility of a
stakeholder who wants to nitpick, i.e., who wants to point at the
possibility that a bad consequence may have been obtained by the
decision maker.

\section{Model Properties}
\label{sec:comparative}
In this section we formally derive the connection of verification and
obfuscation preferences to ambiguity attitudes and comparative statics
of the model.

\begin{definition}[Ambiguity Attitude]
  A decision maker is \emph{ambiguity averse} if for any two acts $f,
  g$ with $f \sim g$, we have $\frac{1}{2} f + \frac{1}{2} g \succsim f$. A
  decision maker is \emph{ambiguity seeking} if for any two acts $f,
  g$ with $f \sim g$, we have $\frac{1}{2} f + \frac{1}{2} g \precsim f$.
\end{definition}

We did not impose any preference for or against ambiguity in our
axioms but merely that for certain binary acts the decision maker is
ambiguity neutral. Nonetheless, it turns out that the axioms are
sufficiently strong to guarantee ambiguity averse behavior
for verification preferences and ambiguity seeking behavior for
obfuscation preferences:
\begin{corollary}[Ambiguity Attitude]\label{coro:ambiguityattitude}
  A decision maker with an expected verification (obfuscation) utility
  is ambiguity averse (seeking).
\end{corollary}
Thus, our expected certification utility model is a special case of
both Choquet expected utility with a supermodular capacity
\parencite{schmeidler_subjective_1989} and maxmin expected utility
\parencite{gilboa_maxmin_1989}.

For performing comparative statics, we first define comparative risk
preferences in the standard way. This will help us to contrast
comparative risk aversion with other comparative results.
\begin{definition}[Comparative Risk Preference]
	$\succsim^2$ is at least as risk averse as $\succsim^1$ if whenever $1/2 x \oplus 1/2 y \succsim^1
		z$, then $1/2 x \oplus 1/2 y \succsim^2 z$. Two decision makers are equally risk averse if
	each is at least as risk averse as the other.
\end{definition}
The following result is a standard result.
\begin{proposition}[Comparative Risk Aversion]\label{prop:risk}
	Suppose $\succsim^1$ and $\succsim^2$ are expected verification utilities or
	expected obfuscation utilities. Then the following statements are equivalent:
	\begin{enumerate}
		\item $\succsim^2$ is at least as risk averse as $\succsim^1$.
		\item $u^2 = t \circ u^1$ for some continuous, concave function $t$.
	\end{enumerate}
\end{proposition}

The following result shows that risk attitude is independent of the attitude
towards verifiable events. Thus, the concern for verifiability cannot be
captured by risk aversion and vice versa. We can see risk aversion as an
aversion to ex-ante uncertainty of perfectly verifiable consequences and preference
for verifiability as an aversion to ex-post uncertainty. Moreover, we only
require identical preferences over two consequences in order to compare two
decision makers with respect to their verifiable events.
\begin{proposition}[ Comparative Statics ] \label{prop:comparative}
	Suppose $\succsim^1$ and $\succsim^2$ are expected verification utilities
	with $\gamma \succ^1 \beta$ and $\gamma \succ^2 \beta$ and identical null events. Then the
	following statements are equivalent:
	\begin{enumerate}
		\item $cl_{\cup}(\mathscr{V}^1) \subseteq cl_{\cup}(\mathscr{V}^2)$.
		\item $\gamma E \beta \sim^2 \gamma (E - F) \beta$ implies $\gamma E \beta \sim^1 \gamma (E - F) \beta$.
	\end{enumerate}
\end{proposition}
The result shows why sensitive events are at the core of preference for
verifiability. If one decision maker's set of sensitive events is a
subset of the other decision maker's, then it is also necessarily the
case that their set of verifiable events is a subset.

In the context of our example, we may care about how inefficient a decision
maker's choice is when following an expected certification utility or an
expected obfuscation utility. A plausible benchmark is expected utility
maximization because the agent is probabilistically sophisticated over
verifiable events.
\begin{definition}[Welfare Loss]
  Let $\succsim$ be an expected verification (obfuscation) utility with
  representation $U,u,\mu,\mathscr{V}$. Let $A \subset \mathscr{A}$ be a
  finite set over which $\succsim$ is strict. The welfare loss is defined as
  \begin{align}
    \label{eq:Distortion}
    L_{\succsim,\mu}(A) =
    \max_{a \in A} \int u \circ a d\mu
    - \max_{b \in \arg \max_{a \in A} U} \int u \circ b d \mu
  \end{align}
\end{definition}

It is noteworthy that from observing expected verification preferences
alone we cannot uniquely determine the welfare loss because the
probabilities of non-verifiable events are not uniquely determined
from behavior. If some consequences are not perfectly verifiable, we
need to complement the preference information with probability
information of non-verifiable events.

Given a fixed probability distribution over states it may be expected
that increases in transparency generally decrease the welfare loss and
align behavior closer with expected utility maximization or that it is
more desirable for a decision maker to be certification-seeking than
obfuscation-seeking. However, policy interventions that change the set
of verifiable events or that change firms' behavior from
obfuscation-seeking to certification-seeking turn out to not
necessarily reduce the welfare loss given a fixed $\mu$. The following
examples show that (holding fixed the verifiable events) a decision
maker acting according to expected certification utility may indeed
face a larger welfare loss than a decision maker acting according to
expected obfuscation utility and that increasing the set of verifiable
events need not lead to a welfare improvement. This is noteworthy
because this implies that calls for better verifiability of CO2
reduction and/or incentives against greenwashing need to be carefully
evaluated for whether they will achieve the desired purpose.

\begin{corollary}[Verification and Obfuscation Welfare Loss
  Comparison]\label{coro:loss}
  Let $\succsim$ be an expected certification utility and $\succsim^{\dagger}$ an
  expected obfuscation utility. Let the risk preferences of $\succsim$ and
  $\succsim^{\dagger}$ be identical. Then either one, but not both, of the
  following statements is true:
  \begin{enumerate}
  \item For all $\mu$ and all decision problems $A$,
    $L_{\succsim,\mu}(A)=L_{\succsim^{\dagger},\mu}(A)=0$.
  \item There exist $\mu$ and decision problems $A$ and $A^{\dagger}$ such
    that $L_{\succsim,\mu}(A) > L_{\succsim^{\dagger},\mu}(A)$ and
    $L_{\succsim^{\dagger},\mu}(A^{\dagger}) > L_{\succsim,\mu}(A^{\dagger})$.

  \end{enumerate}
\end{corollary}
Even though we may intuitively find that behavior according to
expected certification utility is normatively more appealing than that
of expected obfuscation utility, the expected utility loss may be
larger or smaller, depending on the available actions.

\begin{example}
	Suppose in the example decision problem states $s$ and $u$ are very
	unlikely. In this case, an expected utility maximizer would choose
		{\sc RECs}, the same action as a maximizer of expected obfuscation
	utility. A maximizer of expected certification utility would choose
		{\sc Trees} and incur a welfare loss. This is because the decision
	maker fears being blamed for an ex-post suboptimal action that was
	ex-ante optimal.
\end{example}

\begin{definition}[Comparative Loss from Intransparency]
  Let $\succsim$ and $\succsim'$ be expected verification (obfuscation) utility with
  identical risk preferences but different sets of verifiable events
  $\mathscr{V} \subseteq \mathscr{V}'$. Then the welfare loss due to
  intransparency is defined as:
  \begin{align}
    \label{eq:TransparencyLoss}
    T_{\succsim,\mu}(A, \mathscr{V}') = L_{\succsim,\mu}(A) - L_{\succsim',\mu}(A)
  \end{align}
  where $\succsim'$ is the expected verification (obfuscation) utility with
  representation $U,u,\mu,\mathscr{V}'$.
\end{definition}

The following result is obvious, but highlighted for its policy
relevance:
\begin{corollary}[Welfare Loss Compared to Perfect Information]\label{coro:perfect}
  The welfare loss compared to perfect information is always
  nonnegative, $T_{\succsim,\mu}(A, 2^{\mathscr{S}}) \geq 0$.
\end{corollary}
It follows that a policy maker who could influence the set of
verifiable events would always want to implement ex-post certainty
about the consequences. The question is whether more generally any
information increase reduces the welfare loss. Again, it is possible
to find cases in which an information increase may increase the
welfare loss.
\begin{corollary}[Indeterminacy of Welfare Loss] \label{coro:indeterminacy}
  For all $cl_{\cup}\mathscr{V} \subset cl_{\cup}\mathscr{V}' \neq 2^{\mathscr{S}}$, there exist
  decision problems $A$, $A'$ such that
  $T_{\succsim,\mu}(A,\mathscr{V}')<0<T_{\succsim,\mu}(A',\mathscr{V}')$
\end{corollary}

\begin{example}
  In our example, suppose $\mathscr{V}' = \{\{s\}, \{s,t,u\}\}$ and
  $\mathscr{V} = \{\{s,t,u\}\}$. Let $A$ contain only the act {\sc
    Trees} and {\sc RECs} with a slight payoff increase of {\sc RECs}
  in state $u$. Suppose the beliefs are such that $\mu(\{t\})=.99$,
  meaning that an expected utility maximizer would strictly prefer
  {\sc RECs}. Under verifiable events $\mathscr{V}$, the decision
  maker would prefer {\sc RECs} and there is zero welfare loss. In
  $\mathscr{V}'$ the decision maker is able to verify that event
  $\{s\}$ obtains, and now prefers {\sc Trees} over {\sc RECs} because
  in state $s$ the act {\sc Trees} yields a higher payoff. However,
  this leads to a strictly greater welfare loss despite providing
  more information. Together, this and the previous example show that
  policy decisions regarding the incentives for transparency and
  efficiency of carbon emission reduction are nontrivial. While
  transparency and verifiability are perhaps intrinsically desirable,
  they may come at the cost of firms choosing less efficient CO2
  reduction strategies, see \textcite{naef_carbon_2025} for a study on
  afforestation. This suggests that policies that are aimed at
  increasing ex-post verifiability need to be carefully examined
  whether the gain in transparency might be offset by a loss in efficiency ---
  unless, ideally, the policy implements full ex-post transparency.
\end{example}

From these ex-ante counterintuitive results it follows that policies
aimed at increasing but not fully implementing transparency or
inducing certification seeking behavior need to take into account the
available actions of the decision makers. This result already holds
without addressing limiting factors of this welfare analysis; in
general, the decision maker's preferences may not perfectly align with
society's and costs of different policies are not explicitly modeled.

\section{Relation to Literature}\label{sec:literature}

Our model is a special case of Choquet expected utility
\parencite{schmeidler_subjective_1989}. While CEU has lost some of its
popularity as a model of {\em ex-ante} ambiguity aversion,
our interpretation as a preference for verifiability provides a clear
normative motivation for it in the context of {\em ex-post
  uncertainty} and identifies a choice domain in which it is the
natural representation. In particular, as discussed in the axioms
section, comonotonic independence seems a natural starting point for
other, less extreme versions of verifiability preferences than ours.

If the set of verifiable events is exogenously given, a
characterization using cominimum additivity
\parencite{kajii_cominimum_2007,kajii_coextrema_2009} instead of
comonotonic independence is also possible. Cominimum additivity with
respect to a set of events $\mathscr{V}$ requires that all acts that
agree on the worst states in all of the events in $\mathscr{V}$ have
to fulfill the independence axiom. Cominimum additivity provides
intermediate cases between full independence and comonotone
additivity.

Closely related to the representation of the present paper is the dual
self model of ambiguity \parencite{chandrasekher_dual_2022}. In this
model, decision makers maximize an objective $U(a) = \max_{P \in
  \mathscr{P}} \min_{p \in P} \int u \circ a d p$ where $\mathscr{P}$ is a set
of sets of priors. The interpretation of the model is that there are
two selves, a pessimistic one and an optimistic one. The pessimistic
one evaluates the acts according to the maxmin expected utility model
\parencite{gilboa_maxmin_1989} by choosing the worst possible prior.
However, the optimistic self determines the set of available priors
that the pessimist can choose from. Our two main models are special
cases of the dual self model. In a dual self interpretation of our
model, the maximizer chooses a minimal cover of verifiable events.
Each such cover corresponds to a set of priors via the restriction
that a prior $p$ needs to fulfill $p(E)=\mu(E)$ for all events $E$ in
the cover. In other words, to evaluate an act, the minimizer can
choose to assign the entire probability of a verifiable event to the
worst state of the verifiable event. Relative to the characterization
of \parencite{chandrasekher_dual_2022}, our model provides conditions
under which the maximization and minimization steps can be written
inside the expectation.

Preference for verifiability is a form of information preference since
the decision maker cares about how much information is ex post
available about which consequence has been achieved. Information
preferences are often modeled using a two stage approach
\parencite{kreps_temporal_1978,segal_two-stage_1990,dillenberger_additive-belief-based-preferences_2019}.
Given our decision model, the first stage would correspond to the
verifiable events and the second stage to the final consequences and
only the first stage is observed. If verifiable events $\mathscr{V}$
are objectively given and both ex-ante and ex-post probability
distributions over consequences are objective, it is possible to
perform our analysis using preferences over $\Delta \mathscr{V} \times \Delta
\mathscr{X}^{\mathscr{V}}$. However, this way of modeling would
severely restrict the explanatory power of our model. It does not
allow to infer the verifiable events from behavior but requires the
analyst to have data on the ex-post beliefs of the decision maker. In
the context of our example, this would defeat the main purpose of the
paper since ex-post the stakeholder and decision maker would agree on
the probability distribution of the CO2 emissions.

The present paper also provides a decision theoretic foundation for
the standard notion of verifiability employed in contract theory
starting with \textcite{bull_evidence_2004}: verifiable events are
closed under intersection but not necessarily under relative
complements. Our application shows that this notion can be
productively used to model greenwashing.

Our research also relates to definitions
\parencite{defreitas_concepts_2020} and formal models of greenwashing
\parencite{wu_bad_2020} and of green products
\parencite{groening_green_2018}. Unlike previous modeling attempts,
the present paper provides a purely behavioral definition of
greenwashing versus certification-seeking behavior. We do not require
unobservable model components such as incorrect consumer beliefs due
to deception or detailed information about the interaction between
consumers and firms. This yields a very parsimonious model of
greenwashing which does not require any information about the behavior
of consumers of greenwashed products -- preference data of the firm
over ``green'' policies fully identifies the model.

\section{Discussion}
When consequences are not directly observable to decision makers, it
is plausible that these decision makers deviate from expected utility.
We provide a starting point for the analysis of such deviations. As
our axiomatic analysis shows, unobservable consequences can
rationalize a form of ambiguity preferences. Specifically, seemingly
ambiguity averse behavior can be induced from the desire to be able to
prove that good consequences have been reached and seemingly ambiguity
loving behavior can be induced from obfuscation seeking behavior.
These models only provide extreme starting points for a more
general analysis of ex-post verifiability within the CEU framework.

Our reinterpretation of CEU opens a wide field of applications for CEU
when verifiability of consequences matters, such as in environmental
economics. Starting with \textcite{unfccc_race_2023}, there have been
attempts to increasingly rely on CCR to achieve CO2 emission
reductions. These attempts have been criticised as ``inappropriately
verified'' \parencite[][, p.6]{new_corporate_2023}. Our comparative
statics results suggest that ---unless perfect ex-post verifiability
is implemented--- even in our simplistic model application the welfare
effect of policies aimed at increasing verifiability of consequences
are indeterminate. In such cases, traditional environmental policies
(taxation of fossil fuels, energy efficiency mandates, etc.) that
directly incentivize actions rather than consequences avoid this
indeterminacy altogether, since their welfare effects do not depend on
the structure of verifiable events.

\section*{Acknowledgements}
The author gratefully acknowledges support from the Japan Society for
the Promotion of Science (JSPS) under grant number 25K05000.
\newpage

\begin{appendices}
	\renewcommand{\theequation}{\thesection.\arabic{equation}}
\section{Proof of Corollary \ref{coro:ceu}}
We prove sufficiency. By comonotonic independence and the biseparable
utility representation we obtain \parencite[see Proposition 2
of][]{ghirardato_subjective_2003} a representation $U(a) = \int (u \circ a)
d\mu$ where $\mu$ is a capacity and the integral is in the sense of
Choquet.

  Let $m^{\mu}$ be the M\"{o}bius inverse of $\mu$, i.e., $m^{\mu}(E) =
  \sum_{A \subseteq E} (-1)^{|A|-1} \mu(A)$. By Proposition 2 of
  \textcite{chateauneuf_characterizations_1989}, the Choquet integral
  can be
  expressed in terms of the M\"{o}bius inverse \parencite[for a
  comprehensive overview, see ][, p. 235]{grabisch_set_2016}:
	\begin{align}
		\label{eq:ChoquetMobius}
		\int (u \circ a) d\mu = \sum_{E \in \mathscr{E}}m^{\mu}(E)\min_{s \in E} u(a(s)).
	\end{align}
  such that $m^{\mu}(\emptyset)=0$, $\sum_{E \in \mathscr{E}} m^{\mu}(E) = 1$, and
  all partial sums across events contained in an arbitrary event are
  nonnegative.

  Now, define $\mathscr{E}^+ = \{E \in \mathscr{E}| m^{\mu}(E) > 0\}$ and
  $\mathscr{E}^- = \{E \in \mathscr{E}| m^{\mu}(E) < 0\}$. Since $U / x$
  represents $\succsim$ if and only if $U$ represents it, we are allowed to
  divide the representation by an arbitrary factor, in particular by
  $x = \sum_{E \in \mathscr{E}^+} m^{\mu}(E) - \sum_{E \in
    \mathscr{E}^-}m^{\mu}(E)$. Define $p = \sum_{E \in
    \mathscr{E}^+}m^{\mu}(E)/x$ as the probability with which the agent
  will be in a verification requiring situation where the worst
  consequence matters. Further, let $p^+(E) = m^{\mu}(E)/(x \cdot p)$ be the
  conditional probability over what is ex-post verifiable given that
  only the worst consequence matters. Correspondingly, define $p^-(E)
  = - m^{\mu}(E)/(x \cdot (1-p))$ as the conditional probability over what is
  verifiable given that only the best consequence matters. Inserting
  these definitions into \eqref{eq:ChoquetMobius} (divided by $x$) and
  rearranging terms yields the desired result.
\section{Proof of Theorem \ref{thm:verification}}
\begin{proof}
  We prove sufficiency.
  Via the proof of Corollary \ref{coro:ceu} we have a representation
  by a Choquet integral:
	\begin{align}
		\label{eq:ChoquetMobius}
		\int (u \circ a) d\mu = \sum_{E \in \mathscr{E}}m^{\mu}(E)\min_{s \in E} u(a(s)).
	\end{align}

  We first prove that ambiguity neutrality above sensitive events
  implies modularity of $\mu$ above sensitive events.

  \begin{definition}[Modularity above sensitive events]
    A capacity $\mu$ is {\em modular above sensitive events} if for all
    sensitive events $E$ and $F$, and all events $A \subseteq E \cup F$:
	\begin{align}
      \mu(A \cap E \cap F) + \mu(A)
    =
      \mu(A \cap E) +  \mu(A \cap F).
	\end{align}
  \end{definition}

  \begin{lemma}[Modularity]\label{lemm:modularity}
    If $\succsim$ with representation $U$ fulfills proof sufficiency,
    then $\mu$ is modular above sensitive events. In particular,
    sensitive events are closed under unions.
  \end{lemma}
  \begin{proof}
    Without loss of generality, assume $u(\gamma)-u(\beta) = 1 = u(\gamma') - u(\beta')$
    and $f = \gamma A \cap E \beta \sim \gamma' A \cap F \beta' = g$. It follows that
    \begin{align}
      U( 1/2 f \oplus 1/2 g) =
      & 1/2 \mu(A \cap E \cap F) + 1/2 \mu (A) + \frac{u(\beta) + u(\beta')}{2} \\
      U(f) =
      & \mu(A \cap E) + u(\beta) \\
      U(g) =
      & \mu(A \cap F) + u(\beta') \\
    \end{align}
    are all equal by proof sufficiency. But then $2 \cdot U(1/2 f \oplus
    1/2 g) = U(f) + U(g)$ or after substitution:
    \begin{align}
      \mu(A \cap E \cap F) + \mu(A) + u(\beta) + u(\beta') =
      \mu(A \cap E) + u(\beta) + \mu(A \cap F) + u(\beta')
    \end{align}
    and thus $\mu$ is indeed modular above $E \cup F$.

  From modularity of $\mu$ above $E \cup F$ follows that the sensitive
  events are not only closed under intersections, but also under
  unions. This is because modularity implies additivity of the Choquet
  integral above sensitive events $E$ and $F$: $U(1/2 \gamma_{E \cap A} \beta
  \oplus 1/2 \gamma'_{F \cap A} \beta') = 1/2 U(\gamma_{E \cap A} \beta) + 1/2 U(\gamma'_{F \cap A}
  \beta')$.
  \end{proof}

	We now prove the key lemma which relates properties of $m^{\mu}$ to whether a
	set is sensitive.
	\begin{lemma}
		\label{lemm:positiveCritical}
		For all $S \subseteq \mathscr{S}^*$:
		\begin{itemize}
			\item $m(S) \geq 0$, and
			\item if $m(S) > 0$ then $S$ is sensitive and there does not exist a cover of
			      sensitive sets that are all strict subsets of $S$.
		\end{itemize}
	\end{lemma}
	\begin{proof}
		We prove this by induction on the cardinality of $S$, $n$.

		Case $n = 1$: Since $\mu$ is a capacity, $m(S) + m(\emptyset) = \mu(S) \geq \mu(\emptyset) = m(\emptyset) = 0$.
		Thus $m(S) \geq 0$. If $m(S) = \mu(S) > 0 = \mu(\emptyset)$, then since $S$ contains a single
		element, $S$ is sensitive and has no sensitive subsets.

  Case $n > 1:$ Suppose for all sets $S$ of size $n-1$ or smaller the induction
  hypothesis holds. We distinguish the case that $m(S)<0$ from the case that
  $m(S)>0$ and there exists a cover of $S$ of sensitive events that are
  strict subsets of $S$ and derive a contradiction for each case.
  \begin{itemize}
  \item Suppose for sake of contradiction that $m(S) < 0$.

    We first show that $S$ is sensitive and there exists a cover of
    sensitive subsets of $S$. For every $s \in S$ we have by
    monotonicity of the capacity that $\mu(S) \geq \mu(S - \{s\})$. It
    follows from the definition of $m$ that $m(S) + \mu(S- \{s\}) +
    \sum_{E \subset S : s \in E} m(E) = \mu(S) \geq \mu(S-\{s\})$ and thus $m(S) +
    \sum_{E \subset S: s \in E}m(E) \geq 0$. Therefore for some $E \subset S$
    containing $s$ we have that $m(E) > 0$ and by the induction
    hypothesis $E$ is sensitive. It follows that every $s \in S$ is
    contained in a sensitive event and since by Lemma
    \ref{lemm:modularity} sensitive events are closed under unions,
    $S$ is also sensitive.

    Since $n \geq 2$ and there exists a cover of sensitive subsets of $S$ and
    sensitive subsets are closed under unions, we can find a cover $\{A,B\}$ of
    two sensitive subsets of $S$ with $A-B$ and $B-A$ nonempty. Then
    by modularity above sensitive events, $\mu(S) + \mu(A \cap B) = \mu(A) + \mu(B)$. Thus, by the definition
    of $m$, $\sum_{E \subseteq S}m(E) + \sum_{E \subseteq A \cap B} m(E) = \sum_{E \subseteq A}m(E) + \sum_{E \subseteq B}
    m(E)$. This is equivalent to: $m(S) + \sum_{E \subset S: E \not\subseteq A,B}m(E) =
    0$.
    Because by assumption $m(S)<0$, this is only possible if $m(E)>0$ for some
    $E \subset S$, $E \not\subseteq A,B$. But from proof sufficiency follows that
    $\mu(E) + \mu(E \cap A \cap B) = \mu(A \cap E) + \mu(B \cap E)$ and thus $\sum_{F \subset E: F \not\subseteq A, F
      \not\subseteq B} m(F)<0$, contradicting the induction hypothesis.

  \item Suppose for sake of contradiction that $m(S) > 0$ and there
    exists a cover of sensitive subsets of $S$. Then $\mu(S) + \mu(A \cap B) =
    \mu(A) + \mu(B)$ by modularity above sensitive events. Thus, $\sum_{E \subseteq S}m(E) +
    \sum_{E \subseteq A \cap B} m(E) = \sum_{E \subseteq A}m(E) + \sum_{E \subseteq B} m(E)$ which is
    equivalent to: $m(S) + \sum_{E \subset S: E \not\subseteq A,B}m(E) = 0$. But this is
    only possible if $m(E)<0$ for some $E \subset S$, contradicting the
    induction hypothesis.
  \end{itemize}
  \end{proof}

  \begin{lemma}
    \label{lemm:latticePhi}
    If $E,F \in \mathscr{E}$, $m^{\mu}(E) >0$, $m^{\mu}(F) >0$ and $s \in E \cap
    F$, then there exists $G \in \mathscr{E}$ such that $G \subseteq E \cap F$,
    $m^{\mu}(G) > 0$, and $s \in G$.
  \end{lemma}
  \begin{proof}
    By Lemma \ref{lemm:positiveCritical}, $E$ and $F$ are sensitive. It
    follows from proof conjunction that $E \cap F$ is sensitive.
    Suppose $m(G) = 0$ for all $G \subseteq E \cap F$ such that $s \in G$. Then
    $\mu(E \cap F) = \sum_{G \subseteq E \cap F} m^{\mu}(G) = 0 + \sum_{G \subseteq (E \cap F) - \{s\}}m^{\mu}(G) =
    \mu((E \cap F)- \{s\})$, contradicting that $E \cap F$ is sensitive. Thus,
    there exists some $G \subseteq E \cap F$ such that $m^{\mu}(G) > 0$.
 \end{proof}

 \begin{lemma}
   \label{lemm:minimalPhi}
   For every $s \in \mathscr{S}^*$ there exists a unique event $\phi(s) \in
   \mathscr{E}$ such that $m^{\mu}(\phi(s)) >0$ and if $s \in F$ and $\mu(F) >
   0$, then $\phi(s) \subseteq F$.
 \end{lemma}
 \begin{proof}
   Since $s \in \mathscr{S}^*$, it must be included in at least one
   event $E$ with $m^{\mu}(E) >0$. Since $\mathscr{S}$ and $\mathscr{E}$
   are finite, there exists some $E$ such that $m^{\mu}(E) > 0$ and
   $m^{\mu}(F) = 0$ for all $F \subset E$ such that $s \in F$. To see that there
   exists only one such event, suppose $s \in E \cap F$, $m^{\mu}(E)>0$ and
   $m^{\mu}(F)>0$. Then for some $G \subseteq E \cap F$, $m^{\mu}(G) > 0$. If neither
   $E$ nor $F$ have strict subsets on which $m^{\mu}$ is strictly
   positive, then by Lemma \ref{lemm:latticePhi} it must be the case
   that $E = F =G$.
 \end{proof}
 Let $\phi : \mathscr{S}^* \rightarrow \mathscr{E}$ be the function that maps
 states $s$ into the smallest event $E \ni s$ such that $m^{\mu}(E) > 0$.
 Let $\phi^{-1}(E) = \{t \in E | \phi(t) = E\}$ be the set of all states that
 $\phi$ maps into event $E$. By Lemma \ref{lemm:minimalPhi}, $\phi$ and
 $\phi^{-1}$ are well defined.

 We define a probability measure $\eta$ inductively by: $\eta(\{s\}) = 0$ if
 $s \not\in \mathscr{S}^*$, $\eta(\{s\}) = m^{\mu}(\phi(s)) / |\phi^{-1}(\phi(s))|$ if
 $s \in \mathscr{S}^*$, and $\eta(E \cup \{s\}) = \eta(E) + \eta(\{s\})$ for all $E
 \in \mathscr{E}$ and $s \not\in E$. Thus, for every state $s$ we find the
 probability mass of the states $\phi^{-1}(E)$ and divide it evenly among
 these states.

	Denote $\mathscr{V} = \{E \in \mathscr{E} | m^{\mu}(E)>0\}$, then:
	\begin{align}
		\label{eq:finalSummation}
		  & \int_{s \in \mathscr{S}} \max_{E \in \mathscr{V}: s \in E} \min_{t \in E} u(a(t)) d\eta \nonumber     \\
		  & \sum_{s \in \mathscr{S}^*} \eta(s) \max_{E \in \mathscr{V}: s \in E} \min_{t \in E} u(a(t)) \nonumber \\
		= &
		\sum_{s \in \mathscr{S}^*} \eta(s) \min_{t \in \phi(s)} u(a(t)) \nonumber                                 \\
		= &
		\sum_{E \in \mathscr{E}} \sum_{s \in \phi^{-1}(E)} \eta(s) \min_{t \in E} u(a(t)) \nonumber               \\
		= &
		\sum_{E \in \mathscr{E}} m^{\mu}(E) \min_{t \in E} u(a(t))
		=
		U(a).
	\end{align}

	The first equality sign follows from $\eta(s) = 0$ for all $s \not\in \mathscr{S}^*$. The second
	equality sign follows since if $F \ni s$ and $F \in \mathscr{V}$, then by the Lemma
	\ref{lemm:minimalPhi}, $\phi(s) \subseteq F$ and thus $\min_{t \in \phi(s)} u(a(t)) \geq \min_{t
			\in F} u(a(t))$. The third equality sign follows since $\phi: \mathscr{S}^* \rightarrow \mathscr{E}$ is a well
	defined function and thus each state appears exactly once in the summation
	$\sum_{E \in \mathscr{E}} \sum_{s \in \phi^{-1}(E)}$. The fourth equality sign follows by definition
	of $\eta$, since $m^{\mu}(E) = \sum_{s \in \phi^{-1}(E)} \eta(s)$.
\end{proof}
\section{Proof of Corollary \ref{coro:ambiguityattitude}}\label{app:supermodularity}
\begin{proof}
  We first prove the following general lemma.
  \begin{lemma}
    If $\succsim$ is modular above sensitive events and sensitive events are
    closed under unions and intersections, then $\mu$ is supermodular, i.e.,
    for all events $E, F$:
    \begin{align}
      \label{eq:Supermodularity}
      \mu(E \cup F) + \mu(E \cap F) \geq \mu(E) + \mu(F).
    \end{align}
  \end{lemma}
  \begin{proof}
    Note that if an event $E$ is sensitive, removing any state from it
    decreases $\mu(E)$. If it is not sensitive, there exists a state $s$
    such that $\mu(E) = \mu(E-\{s\})$.

    If $E$ and $F$ are sensitive, supermodularity directly follows from
    setting $A = E \cup F$ in the definition of modularity above sensitive
    events.

    If one of the two events is not sensitive because a state $s$ can
    be removed, removing such states from $E-F$ and $F-E$ does not
    change the RHS and does not increase the LHS of
    \eqref{eq:Supermodularity} because $\mu$ is a capacity. Therefore it
    is sufficient to prove supermodularity in the case in which all
    states that can be removed from non-sensitive events $E$ and/or
    $F$ are within $E \cap F$.

    Suppose $E$ is not sensitive and we could remove state $s \in E \cap
    F$. Define $E' = E - \{s\}$. $\mu(E) = \mu(E')$. Thus, again the RHS
    of \eqref{eq:Supermodularity} stays constant if we remove $s$,
    i.e., $\mu(E) + \mu(F) = \mu(E') + \mu(F)$ and the LHS may decrease. Thus,
    supermodularity with respect to $E$ and $F$ is implied by
    supermodularity with respect to $E'$ and $F$. We can iterate this
    step until we reach one of two cases: Either $E$ becomes empty and
    supermodularity holds trivially or $E$ becomes sensitive. If $E$
    becomes sensitive, we repeat the above process for event $F$. At
    the end of this process, either $F$ is empty and $E$ is sensitive,
    or both $E$ and $F$ are sensitive. In the former case,
    supermodularity is trivial because the LHS and RHS are identical
    expressions and in the latter case supermodularity follows from
    setting $A = E \cup F$ in the definition of modularity above
    sensitive events.
  \end{proof}
  Since an expected certification utility satisfies modularity above
  sensitive events (Lemma \ref{lemm:modularity} in the Proof of
  Theorem \ref{thm:verification}) and sensitive events are closed
  under unions and intersections (via proof conjunction and Lemma
  \ref{lemm:modularity}), the capacity $\mu$ is supermodular.

  By \textcite{grabisch_set_2016}, a capacity is supermodular
  if and only if its associated Choquet integral is superadditive,
  i.e., for all bounded functions $h_1, h_2$:
  \begin{align}
    \int (h_1 + h_2)\, d\mu \geq \int h_1\, d\mu + \int h_2\, d\mu.
  \end{align}
  Now let $f \sim g$, i.e., $U(f) = U(g)$. The mixed act $\frac{1}{2} f +
  \frac{1}{2} g$ maps each state $s$ to an outcome such that
  $u\!\left(\frac{1}{2} f(s) \oplus \frac{1}{2} g(s)\right) = \frac{1}{2}
  u(f(s)) + \frac{1}{2} u(g(s))$. Therefore,
  by superadditivity of the Choquet integral:
  \begin{align}
    U\!\left(\tfrac{1}{2} f + \tfrac{1}{2} g\right)
    &= \int \left[\tfrac{1}{2}(u \circ f) + \tfrac{1}{2}(u \circ g)\right] d\mu
    \geq \tfrac{1}{2} \int (u \circ f)\, d\mu + \tfrac{1}{2} \int (u \circ g)\, d\mu \\
    &= \tfrac{1}{2} U(f) + \tfrac{1}{2} U(g) = U(f),
  \end{align}
  and thus $\frac{1}{2} f + \frac{1}{2} g \succsim f$.

  The next lemma is trivial.
  \begin{lemma}[Supermodularity-Submodularity Correspondence]
    $\succsim$ fulfills submodularity if and only if $\succsim^{\dagger}$
    fulfills supermodularity.
  \end{lemma}
  For an expected obfuscation utility $\succsim$, the dual preference
  $\succsim^{\dagger}$ (defined by $a \succsim^{\dagger} b \Leftrightarrow
  b \succsim a$) is an expected certification utility
  (Proof of Theorem \ref{thm:obfuscation}), hence its capacity is supermodular
  by the first part above. By the lemma, the capacity $\mu$ of $\succsim$ is
  submodular.

  By \textcite{grabisch_set_2016}, a capacity is submodular if and only
  if its associated Choquet integral is subadditive, i.e., for all bounded
  functions $h_1, h_2$:
  \begin{align}
    \int (h_1 + h_2)\, d\mu \leq \int h_1\, d\mu + \int h_2\, d\mu.
  \end{align}
  Now let $f \sim g$, i.e., $U(f) = U(g)$. By positive homogeneity and
  subadditivity of the Choquet integral:
  \begin{align}
    U\!\left(\tfrac{1}{2} f + \tfrac{1}{2} g\right)
    &= \int \left[\tfrac{1}{2}(u \circ f) + \tfrac{1}{2}(u \circ g)\right] d\mu
    \leq \tfrac{1}{2} \int (u \circ f)\, d\mu + \tfrac{1}{2} \int (u \circ g)\, d\mu \nonumber\\
    &= \tfrac{1}{2} U(f) + \tfrac{1}{2} U(g) = U(f),
  \end{align}
  and thus $\frac{1}{2} f + \frac{1}{2} g \precsim f$.
\end{proof}
\section{Proof of Proposition \ref{prop:uniqueness}}
\begin{proof}
	$\Leftarrow$ is trivial, we prove $\Rightarrow$. Suppose $\succsim^1 = \succsim^2$.

	From the uniqueness properties of the biseparable preferences
  follows that $U^1 = \theta U^2 + \phi$, $u^1 = \theta u^2 + \phi$ and that the
  associated capacities are identical. What is left to show is the
  relation between the verifiable sets $\mathscr{V}^1$ and
  $\mathscr{V}^2$. Note that the sensitive events in the
  representation are exactly the sets $E$ for which $m^{\mu}(E) > 0$ or
  for which there exists a cover of sets $E_1,\ldots,E_n \subseteq E$ such that $\forall
  i: m^{\mu}(E_i) >0$.

	$cl_{\cup} \mathscr{V}^1 \subseteq cl_{\cup}\mathscr{V}^2$: If $DV =
  cl_{\cup}(\mathscr{V}^1) - cl_{\cup}(\mathscr{V}^2)$ is nonempty, then by
  Lemma \ref{lemm:modularity} (closure of sensitive events under unions)
  for every $V \in DV$, $V$ is sensitive in $U^1$. We show that $V$
  cannot be sensitive in $U^2$: If $V$ is sensitive in $U^2$, then
  either $m^{\mu^2}(V)>0$ and thus $V \in \mathscr{V}^2$ or there exists
  a cover $V_1,\ldots,V_n$ such that $\forall i: m^{\mu^2}(V_i) > 0$. If the
  latter is the case, then $V_1,\ldots,V_n \in \mathscr{V}^2$ and thus $V \in
  cl_{\cup}(\mathscr{V}^2)$. It follows that $V \in DV$ cannot be sensitive
  in $U^2$. But if $V$ is sensitive in $U^1$ but not in $U^2$, then
  the preferences cannot be identical.

	$cl_{\cup} \mathscr{V}^2 \subseteq cl_{\cup}\mathscr{V}^1$ follows by a symmetric
  argument.

	$\forall E \in cl_{\cup}(\mathscr{V}^1): \mu^1(E) = \mu^2(E)$: Suppose for some
  $E \in cl_{\cup}(\mathscr{V}^1)$, $\mu^1(E) > \mu^2(E)$. Then, for some
  $\gamma \succ \beta \succ \alpha$, $\gamma E \beta \succ_1 \alpha$ and $\alpha \succ_2 \gamma E \beta$, yielding a
  contradiction.
\end{proof}

\section{Proof of Theorem \ref{thm:obfuscation}}
\begin{proof}
  We prove sufficiency.
	We define a dual preference $\succsim^{\dagger}$ by: $a \succsim^\dagger b \Leftrightarrow b \succsim
  a$.
	\begin{lemma}[Comonotonic Independence Equivalence]
		$\succsim$ fulfills comonotonic independence if and only if
    $\succsim^{\dagger}$ fulfills comonotonic independence.
	\end{lemma}
	\begin{proof}
		Let $a,b,c$ be comonotonic acts. Then,
		\begin{alignat}{5}
			a                                  & \succsim b                                      &
			\Leftrightarrow^{def.}             &                                                 &
			b                                  & \succsim^{\dagger} a \nonumber                                                                    \\
			                                   & \Leftrightarrow^{com.ind.}                      &   &  &  & \Leftrightarrow^{com.ind.}  \nonumber \\
			\alpha a \oplus (1-\alpha) c       & \succsim \alpha b \oplus (1-\alpha) c \quad     &
			\Leftrightarrow^{def.}             &                                                 &
			\quad \alpha a \oplus (1-\alpha) c & \succsim^{\dagger} \alpha b \oplus (1-\alpha) c
		\end{alignat}
	\end{proof}
	\begin{lemma}[Proof Axiom Correspondence]
		The reactive events fulfill proof sufficiency and proof
    conjunction in $\succsim$ if and only if the sensitive events fulfill
    proof sufficiency and proof conjunction in $\succsim^{\dagger}$.
	\end{lemma}
	\begin{proof}
		$\overline{E}$ is sensitive if $\beta E \gamma \succ^{\dagger} \beta E \cup F \gamma $ for
    all nonnull events $F \subset \overline{E}$. $E$ is reactive in $\succsim$ if $\beta {E \cup
      F} \gamma \succ \beta E \gamma$ for all nonnull events $F \subset \overline{E}$. Thus,
    $\overline{E}$ is sensitive in $\succsim^{\dagger}$ if and only if $E$
    is reactive in $\succsim$.

		Suppose $E$ and $F$ are sensitive in $\succsim^{\dagger}$. Then
    $\overline{E}$ and $\overline{F}$ are reactive in $\succsim$. Thus,
    $\overline{E} \cap \overline{F}$ is reactive in $\succsim$. It follows that
    the complement $E \cap F$ is sensitive in $\succsim^{\dagger}$ and thus
    $\succsim^{\dagger}$ fulfills proof conjunction.

    The equivalence of proof sufficiency then directly follows since
    the symmetric parts of $\succsim$ and $\succsim^{\dagger}$ are identical.
	\end{proof}
	Then $\succsim^{\dagger}$ fulfills comonotonic independence, proof
  conjunction, and proof sufficiency. It follows that $\succsim$ can be
  represented by:
	\begin{align}
		\label{eq:ReverseRepresentation}
		U(a) = & - \int_{s \in \mathscr{S}^*} \max_{E \in \mathscr{V}: s \in E} \min_{s \in E} u(a(s)) d\mu \nonumber     \\
		=      & - \int_{s \in \mathscr{S}^*} \max_{E \in \mathscr{V}: s \in E} - \max_{s \in E} - u(a(s)) d\mu \nonumber \\
		=      & \int_{s \in \mathscr{S}^*} \min_{E \in \mathscr{V}: s \in E} \max_{s \in E} - u(a(s)) d\mu
	\end{align}
  Redefining the utility over consequences yields the desired result.
\end{proof}

\section{Proof of Proposition \ref{prop:risk}}
This result is a standard result. The only interesting question is
whether the definition of comparative risk preference is compatible
with different verifiable events for each decision maker. Note that
the definition of comparative risk preference only depends on
preferences over constant acts. Further, both in expected verification
utility and expected obfuscation utility the definition of the
constant act $1/2 x \oplus 1/2 y$ does not depend on the verifiable events.
Thus, risk attitudes are comparable across decision makers with
different verifiable events and across decision makers with
obfuscation and certification utility.

\section{Proof of Proposition \ref{prop:comparative}}
\begin{proof}
  We prove the result for expected certification utility, the
  derivation for expected obfuscation utility is analogous.

	$1 \Rightarrow 2:$ If for a nonnull $F$, $\gamma E \beta \sim^2 \gamma (E - F) \beta,$ then $E$
  is not sensitive for decision maker 2. Since the sensitive events of
  decision maker 2 contain all sensitive events of decision maker 1,
  $E$ cannot be sensitive for decision maker 1. Thus, $\gamma E \beta \sim^1 \gamma
  (E-F) \beta$.

	$2 \Rightarrow 1:$ If $E$ is sensitive for decision maker 1 but not sensitive
	for decision maker 2, then for decision maker $2$ there exists some
	nonnull event $F \subset E$ such that $\gamma E \beta \sim^2 \gamma (E - F) \beta$. But then also $\gamma E \beta
		\sim^1 \gamma (E - F) \beta$ and thus $E$ cannot be sensitive for decision maker
	1. Thus, the sensitive sets of decision maker 1 are a subset of the
	sensitive sets of decision maker 2.

\end{proof}

\section{Proofs of Corollaries \ref{coro:loss}-\ref{coro:indeterminacy}}
\begin{proof}
  For the proof of Corollary \ref{coro:loss}, consider a decision
  problem $A$ in which the loss given a certification-seeking decision
  maker is nonzero. Then there exists an event $E$ such that no subset
  is verifiable but a superset of which is verifiable. In particular,
  let $F$ and $G$ be disjoint proper subsets of $E$. Consider the acts
  $\gamma F \beta G \delta$ and $\gamma G \beta F \delta'$ with $\gamma \succ \beta$. Assume that the utility
  difference between $\delta$ and $\delta'$ is small. An expected utility
  maximizing DM strictly prefers the former to the latter if and only
  if $\mu(F) > \mu(G)$ (up to the small utility difference between $\delta$ and
  $\delta'$). A certification seeking decision maker strictly prefers the
  former to the latter if and only if $\beta \succ \delta \succ \delta'$. An obfuscation
  seeking decision maker strictly prefers the former to the latter if
  and only if $\delta \succ \delta' \succ \gamma$. Thus, it is possible that the expected
  utility maximizer agrees on the optimal action with the verification
  seeker or the obfuscation seeker without agreeing with the other.

  The proof of \ref{coro:perfect} is trivial, the expected utility of
  the maximizer of $U$ can at most be the maximum expected utility on
  a decision problem. Since if every event is verifiable the decision
  maker maximizes expected utility, the result follows.

  The right inequality of Corollary \ref{coro:indeterminacy} is
  trivial because the decision maker with fewer verifiable events will
  ignore some gains that the other decision maker will not ignore. The
  left inequality follows along the lines of the example provided in
  the main text: for some fixed $\mu$, if $cl_{\cup} \mathscr{V} \subset cl_{\cup}
  \mathscr{V}'$ then there exists an event $E \in \mathscr{V}'$ but not
  in $\mathscr{V}$ and a minimal superset $F \in \mathscr{V}$ of $E$. It
  is then straightforward to construct beliefs and menus of acts for
  which expected utility maximization agrees with $\succsim$ but not $\succsim'$ as
  long as not all states are verifiable.
\end{proof}

\section{Axioms for Biseparable Preferences}
\label{app:biseparable}
In the main text, a biseparable preference is assumed. Here we restate the
axioms used by \textcite{ghirardato_risk_2001} to characterize biseparable
preferences. The axioms are slightly adjusted to fit the notation of the present
paper.
\begin{axiom}[Preference Relation]
  $\succsim$ is a {\em nontrivial preference relation} if it is a complete, transitive,
  nontrivial binary relation.
\end{axiom}

\begin{definition}[Dominant Acts]
  An act $f$ {\em dominates} an act $g$ if for all $s \in \mathscr{S}$, $f(s) \succsim g(s)$.
\end{definition}

\begin{axiom}[Dominance]
  $\succsim$ fulfills {\em dominance} if for all $f,g \in \mathscr{A}$, whenever $f$ dominates $g$,
  then $f \succsim g$.
\end{axiom}

\begin{axiom}[Eventwise Monotonicity]
  $\succsim$ fulfills {\em eventwise monotonicity} if for all $E \in \mathscr{E}$ and $x, y, z \in \mathscr{X}$
  such that $y \succsim z$,
  \begin{itemize}
    \item if $E$ is nonnull, then $x \succ y$ implies $xEz \succ yEz$, and
    \item if $E$ is nonuniversal, then $x \succ z$ implies $yEx \succ yEz$.
  \end{itemize}
\end{axiom}

Let $\mathscr{X}$ be a topological space and let $\mathscr{X}^{\mathscr{S}}$ be endowed with the
product topology.
\begin{axiom}[Continuity]
  $\succsim$ fulfills {\em continuity} if for all nets $\{f_{\alpha}\}_{\alpha \in D} \subseteq X^{\mathscr{S}}$ such
  that $f_{\alpha}$ and $f$ are measurable with respect to the same finite partition,
  whenever $f_{\alpha} \succsim g$ ($g \succsim f_{\alpha}$) for all $\alpha \in D$, then $f \succsim g$, ($g \succsim f$).
\end{axiom}

Continuity ensures the existence of certainty equivalents and
therefore the subjective mixtures
\parencite{nakamura_subjective_1990,ghirardato_subjective_2003,ghirardato_general_2020}
are well defined:
\begin{definition}[Subjective Mixture]
  A {\em subjective mixture} is an act $f \oplus_E g : s \mapsto [{f(s) E g(s)}]$.
\end{definition}

\begin{axiom}[Binary Comonotonic Act Independence]
  $\succsim$ fulfills {\em binary comonotonic act independence} if for all essential
  events $E \in \mathscr{E}$, every $F \in \mathscr{E}$, and for all pairwise comonotonic acts $xEy,
    x'Ey', x''Ey''$, if $x''Ey''$ is dominated by both $xEy$ and $x'Ey'$ or
  dominates both $xEy$ and $x'Ey'$, then
  \begin{align}
    \label{eq:BinaryComonotonicActIndependence}
    xEy \succsim x'Ey' \quad \Rightarrow \quad (x E y) \oplus_B (x'' E y'') \succsim (x'Ey') \oplus_B (x'' E y'').
  \end{align}
\end{axiom}

\begin{theorem}[Biseparable Utility Representation \parencite{ghirardato_risk_2001}]
  Suppose $\succsim$ is a relation on $\mathscr{A}$ and there is at least one essential event $E
    \in \mathscr{E}$, then the following statements are equivalent:
  \begin{enumerate}
    \item $\succsim$ is a nontrivial preference relation that fulfills dominance,
          eventwise monotonicity, continuity, and binary comonotonic act independence.
    \item $\succsim$ is a biseparable preference.
  \end{enumerate}
\end{theorem}

\end{appendices}
\printbibliography
\end{document}